\documentclass{article}

\usepackage{fullpage}
\usepackage{amsthm,amssymb,amsmath}  
\usepackage{xspace,enumerate}
\usepackage[capitalise]{cleveref}
\usepackage[utf8]{inputenc}
\usepackage{thmtools}
\usepackage{thm-restate}
\usepackage{authblk}

  \theoremstyle{plain}
  \newtheorem{theorem}{Theorem}
  \newtheorem{lemma}{Lemma}

  \newtheorem{observation}{Observation}
  \theoremstyle{definition}
  \newtheorem{definition}{Definition}
  \newtheorem{example}{Example}

\title{Palindromic Decompositions with Gaps and Errors}
\author[1]{Michał Adamczyk}
\author[2]{Mai Alzamel}
\author[2]{Panagiotis Charalampopoulos}
\author[2]{Costas S. Iliopoulos}
\author[1,2]{Jakub Radoszewski}

\affil[1]{
Institute of Informatics, University of Warsaw, Warsaw, Poland\\
\texttt{[michal.adamczyk,jrad]@mimuw.edu.pl}
}

\affil[2]{
Department of Informatics, King's College London, London, UK\\
\texttt{[mai.alzamel,panagiotis.charalampopoulos,costas.iliopoulos]@kcl.ac.uk}
}

\date{\vspace{-5ex}}

  \usepackage{amsmath,amsfonts}
  \usepackage{tikz}
  \usepackage[T1]{fontenc}
  \usepackage[utf8]{inputenc}
  \usepackage{comment}
  \usepackage{xspace}
  \usepackage{enumerate}
  \usepackage{listings}
  \usepackage{bibentry}
  \usepackage[hidelinks]{hyperref}
  \usepackage[title]{appendix}
  \usepackage{subfig}
  \def\dd{\mathinner{\ldotp\ldotp}} 

  \usepackage{array}
  \newcolumntype{C}[1]{>{\centering\let\newline\\\arraybackslash\hspace{0pt}}m{#1}}
 
  \usepackage[noend, noline]{algorithm2e}
  \usetikzlibrary{decorations.pathreplacing,calc,decorations.pathmorphing}

  \newcommand{\Oh}{\mathcal{O}}
  \newcommand{\F}{\mathcal{F}}

 \newcommand{\defproblem}[3]{
  \vspace{2mm}
\noindent\fbox{
  \begin{minipage}{0.96\textwidth}
  #1\\
  {\bf{Input:}} #2  \\
  {\bf{Output:}} #3
  \end{minipage}
  }
  \vspace{2mm}
}

\begin{document}
\maketitle

\begin{abstract}  
Identifying palindromes in sequences has been an interesting line of research in combinatorics on words and also in computational biology, after the discovery of the relation of palindromes in the DNA sequence with the HIV virus. Efficient algorithms for the factorization of sequences into palindromes and maximal palindromes have been devised in recent years. We extend these studies by allowing gaps in decompositions and errors in palindromes, and also imposing a lower bound to the length of acceptable palindromes.

We first present an algorithm for obtaining a palindromic decomposition of a string of length $n$ with the minimal total gap length in time $\Oh(n \log{n} \cdot g)$ and space $\Oh(n \cdot g)$, where $g$ is the number of allowed gaps in the decomposition. We then consider a decomposition of the string in maximal $\delta$-palindromes (i.e.~palindromes with $\delta$ errors under the edit or Hamming distance) and $g$ allowed gaps. We present an algorithm to obtain such a decomposition with the minimal total gap length in time $\Oh(n \cdot (g+\delta))$ and space $\Oh(n\cdot g)$.
\end{abstract}

  \section{Introduction}
  A palindrome is a symmetric word that reads the same backward and forward. The detection of palindromes is a classical and well-studied problem in computer science, language theory and algorithm design with a lot of variants arising out of different practical scenarios. 
  String and sequence algorithms related to palindromes have long drawn the attention of stringology researchers~\cite{breslauer1995finding,galil1976real,kolpakov2009searching}. Interestingly, in the seminal Knuth-Morris-Pratt paper presenting the well-known string matching algorithm~\cite{knuth1977fast}, a problem related to palindrome recognition was also considered. In combinatorics on words, for example, studies have investigated the inhabitation of palindromes in Fibonacci words or Sturmian words in general \cite{droubay1995palindromes,DBLP:journals/tcs/DroubayP99}.
There is also an interesting conjecture related to periodicity of infinite strings whose every factor can be decomposed into a bounded number of palindromes~\cite{Frid:2013:PFW:2612801.2612879}.
  
  In computational biology, palindromes play an important role in regulation of gene activity and other cell processes because these are often observed near promoters, introns and specific untranslated regions. Hairpins (also called complemented palindromes) in the HIV virus are strings of the form $x \bar{x}^R$, where $\bar{x}^R$ is the reverse complement of $x$, while (even) palindromes are strings of the form $x x^R$. Algorithms for detecting palindromes can often be adapted to compute hairpins as well. Hence, we can identify palindromes in the DNA sequence and then align the part of the DNA sequence that contains them with the HIV virus.
  
  In the beginnings of algorithmic study of palindromes, Manacher discovered an on-line sequential algorithm that finds all initial palindromes in a string \cite{manacher1975new}. A string $S[1 \dd n]$ is said to have an initial palindrome of length $k$ if $S[1\dd k]$ is a palindrome. Later Apostolico et al.\ observed that the algorithm given by \cite{manacher1975new} is able to find all maximal palindromic factors in the string in $\Oh(n)$ time~\cite{apostolico1995parallel}. Gusfield gave another linear-time algorithm to find all maximal palindromes in a string~\cite{Gusfield:1997:AST:262228}. He also discussed the relation between biological sequences and gapped (separated) palindromes (i.e. strings of the form $x v\bar{x}^R$).
 Gupta et al.~\cite{gupta2015searching} presented an $\Oh(n)$-time algorithm to compute specific classes---based on length constraints---of such palindromes. Algorithms for finding the so-called gapped palindromes were also considered in \cite{DBLP:conf/iwoca/FujishigeNIBT16,kolpakov2009searching}. (In our study, we consider gaps \emph{between} palindromes, not inside them.)
  
  A problem that gained significant attention recently was decomposing a string into a minimal number of palindromes; any such decomposition is called a palindromic factorization. Fici et al.\ \cite{Fici:2014:SAM:2953214.2953656} presented an on-line $\Oh(n \log n)$-time algorithm for computing a palindromic factorization of a string of length $n$. A similar on-line algorithm was presented by I et al.\ \cite{I2014} as well as an on-line algorithm with the same time complexity to factorize a string into maximal palindromes. Alatabbi et al.\ gave an off-line $\Oh(n)$-time algorithm for the latter problem~\cite{alatabbi2013maximal}. In addition, Rubinchik and Shur~\cite{DBLP:conf/iwoca/RubinchikS15} devised an $\Oh(n)$-sized data structure that helps locate palindromes in a string; they also show how it can be used to compute the palindromic factorization of a string in $\Oh(n \log n)$ time.

A similar problem, first studied by Galil and Seiferas in~\cite{Galil:1978:LOR:322047.322056}, asked whether a given string  can be decomposed into $k$ palindromes. Galil and Seiferas~\cite{Galil:1978:LOR:322047.322056} presented an on-line $\Oh(n)$-time algorithm for $k=1,2$ and an off-line $\Oh(n)$-time algorithm for $k=3,4$. In 2014, Kosolobov et al.\ presented an on-line $\Oh(kn)$-time algorithm to decide this for arbitrary $k$~\cite{DBLP:conf/sofsem/KosolobovRS15}. 
 
  Our work is a continuation of this line of research, motivated by possible errors and inconsistencies in the biological data. We extend the previous work by introducing a constraint on the length of the palindromes and allowing gaps and errors in the decompositions. By \emph{gaps} we mean regions of the string that are not decomposed into palindromes of sufficient length. We allow \emph{errors} in the palindromes, so that a \emph{palindrome with errors} is a string having a small Hamming or edit distance from an ideal palindrome. We present two approaches for decomposing a string into sufficiently long palindromes; one allowing only gaps in the decomposition and the other allowing both gaps in the decomposition and errors in the palindromes. We first present an algorithm that computes a palindromic decomposition with the minimal total gap length of a string of length $n$ in time $\Oh(n \log{n} \cdot g)$ and space $\Oh(n \cdot g)$, where $g$ is the number of allowed gaps. Secondly, we present an $\Oh(n \cdot (g + \delta))$-time and $\Oh(n \cdot g)$-space algorithm for the decomposition of a string of length $n$ into maximal palindromes with at most $\delta$ errors each, under the Hamming or edit distance, and $g$ allowed gaps. The algorithms can be applied for both standard and complemented palindromes.

Our first result is a multistage generalization of the algorithm computing a palindromic factorization from \cite{Fici:2014:SAM:2953214.2953656}. We show that their approach can be extended to handle gaps, a constraint on the length of palindromes, and complemented palindromes as well. In our second result, the first step is the computation of maximal palindromes with at most $\delta$ errors each and the second step is a dynamic programming approach that takes as input these maximal palindromes and returns a decomposition of the string with the minimal total gap length. An $\Oh(n \cdot \delta)$ computation of maximal palindromes under the Hamming distance was presented in \cite{Gusfield:1997:AST:262228}. We present an equally efficient approach for computing all maximal palindromes under the edit distance. After our conference publication \cite{CSR2017}, we have learnt about alternative algorithms computing this type of maximal approximate palindromes \cite{DBLP:journals/pr/PortoB02,DBLP:journals/ijfcs/HsuCC10}. The approach of \cite{DBLP:journals/pr/PortoB02} works in $\Oh(n \cdot \delta^2)$ time. The algorithm from \cite{DBLP:journals/ijfcs/HsuCC10} has same time complexity as ours ($\Oh(n \cdot \delta)$), however, it is much more complex, as it was designed for a more general problem.

\section{Notation and terminology}
Let $S=S[1] S[2]\cdots S[n]$ be a \textit{string} of \textit{length} $|S|=n$ over an alphabet $\Sigma$. We consider the case of an integer alphabet; in this case each letter can be replaced by its rank so that the resulting string consists of integers in the range $\{1,\ldots,n\}$. For two positions $i$ and $j$, where $1 \leq i \leq j \leq n$, in $S$, we denote the \textit{factor} $S[i]S[i+1]\cdots S[j]$ of $S$ by $S[i\dd j]$.
We denote the reverse string of $S$ by $S^R$, i.e.~$S^R=S[n]S[n-1] \cdots S[1]$. The empty string (denoted by $\varepsilon$) is the unique string over $\Sigma$ of length 0. A string $S$ is said to be a \emph{palindrome} if and only if $S=S^R$.
If $S[i\dd j]$ is a palindrome, the number $\frac{i+j}{2}$ is called the center of $S[i\dd j]$. Let $S[i \dd j]$, where $1 \leq i \leq j \leq n$, be a palindromic factor in $S$. It is said to be a \emph{maximal palindrome} if there is no longer palindrome in $S$ with center $\frac{i+j}{2}$.  Note that a maximal palindrome can be a factor of another palindrome.

\begin{definition}
We say that $S=p_1 p_2 \cdots p_{\ell}$ is a (maximal) palindromic decomposition of $S$ if all the strings $p_i$ are (maximal) palindromes. 
\end{definition}

\begin{definition}
A (maximal) palindromic decomposition of $S$ such that the number of (maximal) palindromes is minimal is called a (maximal) palindromic factorization of $S$.
\end{definition}

Note that any single letter is a palindrome and, hence, every string can always be decomposed into palindromes. However, not every string can be decomposed into maximal palindromes; e.g. consider $S=\texttt{abaca}$~\cite{alatabbi2013maximal}. 

Let $f$ be an \emph{involution} on the alphabet $\Sigma$, i.e., a function such that $f^2=\mathrm{id}$. We extend $f$ into a morphism on strings over $\Sigma$. We say that a string $x$ is a \emph{generalized palindrome} if $x=f(x^R)$. Two known notions fit this definition:
  \begin{itemize}
    \item If $f=\mathrm{id}$, then a generalized palindrome is a standard palindrome.
    \item If $\Sigma=\{\mathtt{A},\mathtt{C},\mathtt{G},\mathtt{T}\}$ and
    $f(\mathtt{A})=\mathtt{T}$, $f(\mathtt{C})=\mathtt{G}$, $f(\mathtt{G})=\mathtt{C}$, $f(\mathtt{T})=\mathtt{A}$, then a generalized palindrome corresponds to a so-called complemented palindrome~\cite{Gusfield:1997:AST:262228}.
  \end{itemize}
  
  \begin{example}
  The string $\mathtt{A~G~T~A~C~T~T~C~A~T~G~A}$ is a standard palindrome and the string $\mathtt{T~A~G~T~C~G~A~C~T~A}$ is a complemented palindrome.
  \end{example}
  
  We also consider (generalized) palindromes with errors.
  Let us recall two well-known metrics on strings.
  Let $u$ and $v$ be two strings.
  If $|u|=|v|$, then the \emph{Hamming distance} between $u$ and $v$ is the number of positions where $u$ and $v$ do not match.
  The \emph{edit (or Levenshtein) distance} between $u$ and $v$ is the minimum number of edit operations (insertions, deletions, substitutions) needed to transform $u$ into $v$.
  We say that $x$ is a \emph{generalized $\delta$-palindrome} under the Hamming distance (or the edit distance) if
  the minimum Hamming distance (edit distance, respectively) from $x$ to any generalized palindrome is \emph{at most} $\delta$.
  
  A generalized palindrome $S[i\dd j]$ is called \emph{maximal}
  if there is no longer generalized palindrome with the same center.
  Similarly, a generalized $\delta$-palindrome $S[i\dd j]$ under the Hamming/edit distance is called \emph{maximal} if there is no longer generalized $\delta$-palindrome under the same distance measure with the same center.

\begin{example}
 All maximal 0-palindromes/1-palindromes in $\mathtt{GTATCG}$ (for $f = \mathrm{id}$) under the Hamming and under the edit distance are as follows: 

	\begin{center}
        \begin{tabular}    {|C{1.5cm}|C{0.8cm}|C{0.8cm}|C{0.8cm}|C{0.8cm}|C{1cm}|C{1.2cm}|C{0.8cm}|C{0.8cm}|C{0.8cm}|C{0.8cm}|C{0.8cm}|}
\hline
center & 1 & 1.5 & 2 & 2.5 & 3 & 3.5 & 4 &4.5&5&5.5&6\\
\hline
0& $\mathtt{G}$ & $\mathtt{\varepsilon}$ & $\mathtt{T}$ & $\mathtt{\varepsilon}$ & $\mathtt{TAT}$ & $\mathtt{\varepsilon}$ & $\mathtt{T}$ & $\mathtt{\varepsilon}$ & $\mathtt{C}$ & $\mathtt{\varepsilon}$ & $\mathtt{G}$
\\ \hline
1 under Hamming&$\mathtt{G}$ & $\mathtt{GT}$ & $\mathtt{GTA}$ & $\mathtt{TA}$ & $\mathtt{GTATC}$ & $\mathtt{AT}$ &$\mathtt{ATC}$&$\mathtt{TC}$&$\mathtt{TCG}$&$\mathtt{CG}$&$\mathtt{G}$
\\ \hline
1 under edit&$\mathtt{G}$ & $\mathtt{GT}$ & $\mathtt{GTA}$ & $\mathtt{GTAT}$ & $\mathtt{GTATC}$ & $\mathtt{GTATCG}$ & $\mathtt{ATC}$ & $\mathtt{TC}$ &$\mathtt{TCG}$&$\mathtt{CG}$&$\mathtt{G}$
\\ \hline
\end{tabular}
	\label{tab:ex5}
	\end{center}
For instance, the whole string \texttt{GTATCG} is a 1-palindrome under the edit distance, as deleting its fifth letter yields a palindrome \texttt{GTATG}.
\end{example}

The computational problems we study can be formally stated as follows.

{\defproblem{\textsc{Generalized Palindromic Decomposition with Gaps}}{A string $S$ of length $n$, an involution $f$, and integers $g,m \geq 1$}{A decomposition of $S$ into generalized palindromes with the minimal possible total length of gaps, $\sum_{i}^{q} |g_i|$, such that:
\begin{itemize}
\item There are at most $g$ gaps, i.e. $q \le g$
\item Each palindrome is of length at least $m$
\end{itemize}}}
  
{\defproblem{\textsc{Generalized Maximal $\delta$-Palindromic Decomposition with Gaps}}{A string $S$ of length $n$, an involution $f$, and integers $g,m,\delta \geq 1$}{A decomposition of $S$ into maximal generalized $\delta$-palindromes with the minimal possible total length of gaps, $\sum_{i}^{q} |g_i|$, such that:
\begin{itemize}
\item There are at most $g$ gaps, i.e. $q \le g$
\item Each generalized $\delta$-palindrome is of length at least $m$
\end{itemize}}}

We apply several instances of dynamic programming. For simplicity of presentation, we only show how to compute the minimal total length of gaps and omit describing the retrieval of the decomposition itself. To compute the latter, in each of the dynamic programming matrices we would store a pointer to the cell that gave us the minimum value so that we could actually compute the decomposition with the minimal total length of the gaps by backtracing.
  
\section{Palindromic decomposition with gaps}\label{sec:first}
In this section we develop an efficient solution to the \textsc{Generalized Palindromic Decomposition with Gaps} problem. It is based on several transformations of the algorithm for computing a palindromic factorization by Fici et al.\ \cite{Fici:2014:SAM:2953214.2953656}. For a string $S$ of length $n$ this algorithm works in $\Oh(n \log n)$ time. The algorithm consists of two steps:
  \begin{enumerate}
  \item Let $P_j$ be the sorted list of starting positions of all palindromes ending at position $j$ in $S$.
 This list may have size $\Oh(j)$. However, it follows from combinatorial properties of palindromes that the sequence of consecutive differences in $P_j$ is non-increasing and contains at most $\Oh(\log j)$ distinct values.
Let $P_{j,\Delta}$ be the maximal sublist of $P_j$ containing elements whose predecessor in $P_j$ is smaller by exactly $\Delta$.
Then there are $\Oh(\log j)$ such sublists in $P_j$.
 Hence, $P_j$ can be represented by a set $G_j$ of size $\Oh(\log j)$ which consists of triples of the form $(i,\Delta,k)$ that represent $P_{j,\Delta} = \{i,i+\Delta,\ldots,i+(k-1)\Delta\}$. 
The triples are sorted according to decreasing values of $\Delta$ and all starting positions in each triple are greater than in the previous one.
Fici et al.\ show that $G_j$ can be computed from $G_{j-1}$ in $\Oh(\log j)$ time.
 \item Let $PL[j]$ denote the number of palindromes in a palindromic factorization of $S[1\dd j]$.
Fici et al.\ show that it can be computed via a dynamic programming approach, using all palindromes from $G_j$ in $\Oh(\log j)$ time.
Their algorithm works as follows.
Let $PL_\Delta[j]$ be the minimum number of palindromes we can decompose $S[1 \dd j]$ in, provided that we use a palindrome from $(i,\Delta,k) \in G_j$.
Then $PL_\Delta[j]$ can be computed in constant time using $PL_\Delta[j-\Delta]$ based on the fact that if $(i,\Delta,k) \in G_j$ and $k \geq 2$, then $(i,\Delta,k-1) \in G_{j-\Delta}$.
Exploiting this fact, $PL_\Delta[j]$ can be computed by only considering $PL_\Delta[j-\Delta]$ and the shortest palindrome in $(i,\Delta,k)$. Finally, we compute $PL[j]$ from all such $PL_\Delta[j]$ values.
\end{enumerate}

In Appendix~\ref{app} we show for completeness that the same approach works for generalized palindromes for any involution $f$.

  To solve the \textsc{Generalized Palindromic Decomposition with Gaps} problem, we first need to modify each of the triples in $G_j$ to reflect the length constraint ($m$). More precisely, due to the length constraint, in each $G_j$ some triples will disappear completely, and at most one triple will get \emph{trimmed} (i.e. the parameter $k$ will be decreased).

  Our algorithm then computes an array $MG[1\dd n][0\dd g]$ such that $MG[j][q]$ is the minimum possible total length of gaps in a palindromic decomposition of $S[1\dd j]$, provided that there are at most $q$ gaps. Simultaneously, our algorithm computes an auxiliary array $MG'[1\dd n][0\dd g]$ such that $MG'[j][q]$ is the minimum possible total length of gaps up to position $j$ provided that this position belongs to a gap: at most the $q$-th one. 

  For $j>0$ and $q \ge 0$ we have the following formula:
  $$MG[j][q] = \min(MG'[j][q],\min_\Delta\{MG_\Delta[j][q]\})$$
  where $MG_\Delta[j][q]$ is the partial minimum computed only using generalized palindromes from $(i,\Delta,k) \in G_j$. The formula means: either we have a gap at position $j$, or we use a generalized palindrome ending at position $j$. We also set $MG[0][q]=0$ for any $q \ge 0$.

   We compute $MG_\Delta[j][q]$ for $(i,\Delta,k) \in G_j$ using the same approach as Fici et al.\ \cite{Fici:2014:SAM:2953214.2953656} used for $PL_\Delta$, ignoring the triples that disappear due to the length constraint. If there is a triple that got trimmed, then the corresponding triple at position $j-\Delta$ (from which we reuse the values in the dynamic programming) must have got trimmed as well. More precisely, if the triple $(i, \Delta, k)$ is trimmed to $(i, \Delta, k')$ at position $j$, then at position $j-\Delta$ there is a triple $(i, \Delta, k-1)$ which is trimmed to $(i, \Delta, k'-1)$; that is, by the same number of generalized palindromes. Consequently, to compute $MG_\Delta[j][q]$ from $MG_\Delta[j-\Delta][q]$, we need to include one additional generalized palindrome (the shortest one in the triple) just as in Fici et al.'s approach.

\begin{example}
Consider the string \texttt{AACCAACCAACCAACCAA}, $f=\mathrm{id}$, and let $m=7$.

  \begin{center}
  \begin{tikzpicture}[xscale=0.6]
    \foreach \x/\c in {1/A,2/A,3/C,4/C,5/A,6/A,7/C,8/C,9/A,10/A,11/C,12/C,13/A, 14/A, 15/C, 16/C, 17/A, 18/A}{
      \draw (\x,0) node[above] {\tt \c};
    }
    \foreach \x in {1,...,18} \draw (\x,0) node[below] {\scriptsize \x};
  \end{tikzpicture}
  \end{center}

Then $G_{18}=\{(1, \infty,  1), (5,4,4), (18,1,1) \}$, where:
\begin{itemize}
  \item $(1, \infty,  1)$ represents the whole string,
  \item $(5,4,4)$ represents $\{ \texttt{AACCAACCAACCAA}, \texttt{AACCAACCAA}, \texttt{AACCAA}, \texttt{AA} \}$ which will get trimmed by 2 palindromes due to the length constraint, becoming $(5,4,2)$,
  \item $(18,1,1)$ represents $\{\texttt{A}\}$ and disappears.
\end{itemize}
Now looking at position $j-\Delta=18-4=14$ for the trimmed group, we had $(5,4,3) \in G_{14}$ representing $\{ \texttt{AACCAACCAA}, \texttt{AACCAA}, \texttt{AA} \}$, and this also gets trimmed by $2$ palindromes, becoming $(5,4,1)$.

\end{example}
  
  Finally, for $j>0$ and $q>0$ we compute $MG'$ using the following formula:
  $$MG'[j][q] = \min(MG'[j-1][q],MG[j-1][q-1])+1.$$
  The first case corresponds to continuing the gap from position $j$, whereas the second to using a generalized palindrome finishing at position $j-1$ or a gap finishing at position $j-1$ (the latter will be suboptimal). Here the border cases are $MG'[j][0] = \infty$ for $j \ge 0$ and $MG'[0][q] = \infty$ for $q>0$.
  
  Thus we arrive at the complete solution to the problem.
  
\begin{theorem}
The \textsc{Generalized Palindromic Decomposition with Gaps} problem can be solved in $\Oh(n \log n \cdot g)$ time and $\Oh(n\cdot g)$ space.
\end{theorem}

\section{Computing maximal palindromes with errors}
  Recall that all maximal (standard) palindromes in a string can be computed in $\Oh(n)$ time by Manacher's~\cite{crochemore2003jewels,manacher1975new} and Gusfield's~\cite{Gusfield:1997:AST:262228} algorithms. These algorithms perform different computations for odd- and for even-length palindromes. Recall that we defined the centers of odd-length palindromes as integers and the centers of even-length palindromes as odd multiples of $\frac12$.
  
  Gusfield's algorithm \cite{Gusfield:1997:AST:262228} applies Longest Common Extension (LCE) Queries in the string $T=S \$ S^R$, where $\$ \not\in\Sigma$ is a sentinel character. An $LCE(i,j)$ query returns the length of the longest common prefix of the suffixes $T[i\dd |T|]$ and $T[j\dd |T|]$. For example, to compute the length of the maximal even-length palindrome centered between positions $i$ and $i+1$, the algorithm computes $LCE(i+1,2n+2-i)$ in $T$. Recall that LCE queries in a string (over an integer alphabet) can be answered in $\Oh(1)$ time after linear-time preprocessing~\cite{AlgorithmsOnStrings}.
  
  Gusfield's approach can be easily adapted to generalized palindromes: it suffices to apply LCE-queries on $T=S\$f(S^R)$. To further simplify the description of this approach, we introduce the Longest Gapped Palindrome (LGPal) Queries, such that $LGPal(i,j)$ is the maximum $k$ such that $f(S[i-k+1\dd i]^R) = S[j\dd j+k-1]$; see Fig.~\ref{fig:LCPal}. As we have already noticed, LGPal-queries are equivalent to LCE-queries in $T=S\$f(S^R)$.
  
  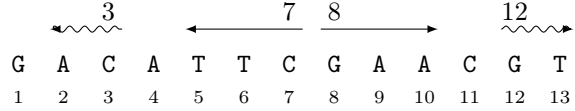
\begin{figure}
  \begin{center}
  \begin{tikzpicture}[xscale=0.6]
    \foreach \x/\c in {1/G,2/A,3/C,4/A,5/T,6/T,7/C,8/G,9/A,10/A,11/C,12/G,13/T}{
      \draw (\x,0) node[above] {\tt \c};
    }
    \foreach \x in {1,...,13} \draw (\x,0) node[below] {\scriptsize \x};
    \draw[latex-,decorate,decoration={snake,amplitude=.4mm,segment length=2mm}] (1.7,0.7) -- (3.3,0.7);
    \draw[-latex,decorate,decoration={snake,amplitude=.4mm,segment length=2mm}] (11.7,0.7) -- (13.3,0.7);
    \draw[latex-] (4.7,0.7) -- (7.3,0.7);
    \draw[-latex] (7.7,0.7) -- (10.3,0.7);
    \foreach \x in {3,7,8,12} \draw (\x,0.7) node[above] {\x};
  \end{tikzpicture}
  \end{center}
  \caption{To find the longest complemented 1-palindrome under the Hamming distance centered at position $7.5$ in $S=\mathtt{GACATTCGAACGT}$, it suffices to ask two LGPal-queries: $LGPal(7,8)=3$ finds the first mismatch, and $LGPal(3,12)$ extends the 1-palindrome after the mismatch. Note that each of these LGPal-queries is equivalent to an appropriate LCP-query in $S\$f(S^R)$.
  \label{fig:LCPal}}
  \end{figure}

  It is known (see~\cite{Gusfield:1997:AST:262228}) that all maximal generalized $\delta$-palindromes under the Hamming distance can be computed in $O(n \cdot \delta)$ time via at most $\delta$ applications of the LGPal-query for each possible center position. Below we show how to compute maximal generalized $\delta$-palindromes under the edit distance within the same time complexity.
  
  Recall that if $u$ is a generalized $\delta$-palindrome under the edit distance, then there exists a generalized palindrome $v$ such that the minimal number of edit operations (insertion, deletion, substitution) required to transform $u$ to $v$ is at most $\delta$. The following simple observation shows that we can restrict ourselves to deletions and substitutions only, which we call in what follows the \emph{restricted edit operations}. Intuitively, instead of inserting at position $i$ a character to match the character at position $|u|-i+1$, we can delete the character at position $|u|-i+1$.
  
  \begin{observation}
    Let $u$ be a generalized $\delta$-palindrome and $v$ a generalized palindrome such that the edit distance between $u$ and $v$ is minimal. Then there exists a generalized palindrome $v'$ such that the number of restricted edit operations needed to transform $u$ to $v'$ is equal to the edit distance between $u$ and $v$.
  \end{observation}
  
  We can extend a maximal generalized $\delta$-palindrome $S[i\dd j]$ to a maximal generalized $(\delta+1)$-palindrome in three ways; either ignore the letter $S[i-1]$ and then perform an LGPal-query, or ignore the letter $S[j+1]$ and then perform an LGPal-query, or ignore both and then perform the LGPal-query. More formally:
  
  \begin{definition}
  Assume that $S[i\dd j]$ is a generalized $\delta$-palindrome.
  Then we say that each of the factors $S[i'\dd j']$ for:
  \begin{itemize}
    \item $i'=i-1-d$, $j'=j+d$, where $d=LGPal(i-2,j+1)$
    \item $i'=i-d$, $j'=j+1+d$, where $d=LGPal(i-1,j+2)$
    \item $i'=i-1-d$, $j'=j+1+d$, where $d=LGPal(i-2,j+2)$
  \end{itemize}
  is an \emph{extension} of $S[i\dd j]$.
  If the index $i'$ is smaller than 1 or the index $j'$ is greater than $|S|$,
  the corresponding extension is not possible.
  We also say that $S[i\dd j]$ can be extended to any of the three strings $S[i'\dd j']$.
  \end{definition}
  Clearly, the extensions of a generalized $\delta$-palindrome are always generalized $(\delta+1)$-palindromes.
  
  To facilitate the case of $\delta$-palindromes being prefixes or suffixes of the text, we also introduce the following \emph{border-reductions} for $S[i\dd j]$ being a generalized $\delta$-palindrome:
  \begin{itemize}
    \item If $i=1$, a border reduction leads to $S[1\dd j-1]$.
    \item If $j=n$, a border reduction leads to $S[i+1\dd n]$.
  \end{itemize}
  If any of the reductions is possible, we also say that $S[i \dd j]$ can be border-reduced to the corresponding strings. As previously, border-reductions of a generalized $\delta$-palindrome are always generalized $(\delta+1)$-palindromes.
  
  \begin{lemma}\label{lem:ext}
    Given a maximal generalized $\delta$-palindrome $S[i' \dd j']$ with $\delta>0$, there exists a maximal generalized $(\delta-1)$-palindrome $S[i \dd j]$ which can be extended or border-reduced to $S[i' \dd j']$.
  \end{lemma}
  \begin{proof}
    Consider a shortest sequence of restricted edit operations that transforms $u=S[i'\dd j']$ into a generalized palindrome $v$. Let us consider the position where we perform a restricted edit operation that is closest to $i'$ or $j'$. Assume w.l.o.g.\ that this position---denote it by $e$---is not further to $i'$ than to $j'$.
    
    Assume first that this edit operation is a substitution. Then $S[i\dd j]$, for $i=e+1$ and $j=j'-(e+1-i')$, is a generalized $(\delta-1)$-palindrome (the witness generalized palindrome is the corresponding factor of $v$); see Fig.~\ref{sub}. Moreover, it is a maximal generalized $(\delta-1)$-palindrome, as otherwise $S[e]=S[i-1]$ would be equal to $f(S[j+1])$, which means that the substitution at the position $e$ would not be necessary. This completes the proof in this case.

  \begin{figure}
  \begin{center}
  \begin{tikzpicture}[xscale=0.6]
    \foreach \x/\c in {1/i',2/,3/,4/e,5/i,6/,7/,8/,9/j,10/,11/,12/,13/j'}{
      \draw (\x,0) node[above] {$\c$};
    }
    \draw[latex-,decorate,decoration={snake,amplitude=.4mm,segment length=2mm}] (0.8,0.7) -- (3.3,0.7);
    \draw[-latex,decorate,decoration={snake,amplitude=.4mm,segment length=2mm}] (10.5,0.7) -- (13,0.7);
    \draw[latex-] (4.7,0.7) -- (7.7,0.7);
    \draw[-latex] (7.7,0.7) -- (9.1,0.7);
    \draw (4,0.45) node[above] {$X$};
    \draw (9.8,0.45) node[above] {$Y$};
  \end{tikzpicture}
  \end{center}
  \caption{If the outermost restricted edit operation on $S[i' \dd j']$ is a substitution (from letter $X$ to letter $Y$), then $S[i' \dd j']$ is an extension of the third type of the maximal generalized $(\delta-1)$-palindrome $S[i\dd j]$.}
  \label{sub}
  \end{figure}
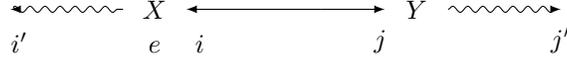
    
    Now assume that the edit operation at the position $e$ was a deletion. Let $a=e+1$ and $b=j'-(e-i')$. Again, we see that clearly $S[a\dd b]$ is a generalized $(\delta-1)$-palindrome. If it is maximal, then we are done. Otherwise, consider the maximal generalized $(\delta-1)$-palindrome $S[i\dd j]$ centered at the same position as $S[a\dd b]$ ($a-i = j-b >0$). Now we have three cases; see Fig.~\ref{del}.  
    
  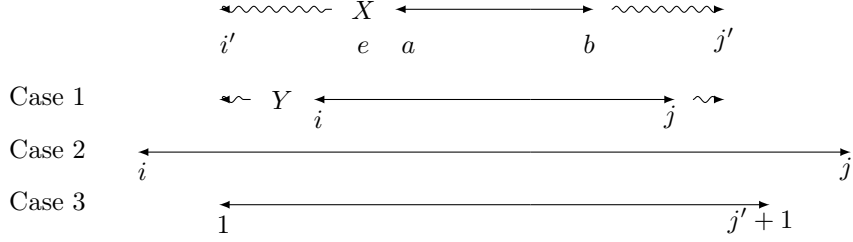
\begin{figure}
  \begin{center}
  \begin{tikzpicture}[xscale=0.6]
    \foreach \x/\c in {1/i',2/,3/,4/e,5/a,6/,7/,8/,9/b,10/,11/,12/j'}{
      \draw (\x,0) node[above] {$\c$};
    }
    \draw[latex-,decorate,decoration={snake,amplitude=.4mm,segment length=2mm}] (0.8,0.7) -- (3.3,0.7);
    \draw[-latex,decorate,decoration={snake,amplitude=.4mm,segment length=2mm}] (9.5,0.7) -- (12,0.7);
    \draw[latex-] (4.7,0.7) -- (7.7,0.7);
    \draw[-latex] (7.7,0.7) -- (9.1,0.7);
    \draw (4,0.45) node[above] {$X$};
    
    \draw (-3,-0.7) node[above] {Case 1};
    \draw[latex-] (2.9,-0.5) -- (7.7,-0.5);
    \draw[-latex] (7.7,-0.5) -- (10.9,-0.5);
    \draw (3,-1) node[above] {$i$};
    \draw (10.8,-1) node[above] {$j$};    
        \draw (2.2,-0.75) node[above] {$Y$};
    \draw[latex-,decorate,decoration={snake,amplitude=.4mm,segment length=2mm}] (0.8,-0.5) -- (1.5,-0.5);
    \draw[-latex,decorate,decoration={snake,amplitude=.4mm,segment length=2mm}] (11.3,-0.5) -- (12,-0.5);
    
    \draw (-3,-1.4) node[above] {Case 2};
    \draw[latex-] (-1,-1.2) -- (7.7,-1.2);
    \draw[-latex] (7.7,-1.2) -- (14.8,-1.2);
    \draw (-0.9,-1.7) node[above] {$i$};
    \draw (14.7,-1.7) node[above] {$j$};    
    
    \draw (-3,-2.1) node[above] {Case 3};
    \draw[latex-] (0.8,-1.9) -- (7.7,-1.9);
    \draw[-latex] (7.7,-1.9) -- (13,-1.9);
    \draw (0.9,-2.4) node[above] {1};
    \draw (12.8,-2.4) node[above] {$j'+1$};
  \end{tikzpicture}
  \end{center}
  \caption{Three cases resulting when the outermost edit operation on $S[i'\dd j']$ is a deletion of a character $X$.}
  \label{del}
  \end{figure}

    \begin{enumerate}
    \item If $j \le j'$, then we can obtain $S[i'\dd j']$ by an extension (of the first type) of $S[i\dd j]$; i.e.\ ignoring the letter $S[i-1]$.
    \item If $j > j'$, then we have that $S[i'\dd j'+1]$ is a generalized $(\delta-1)$-palindrome. If, additionally, $i' > 1$, then $S[i'-1\dd j'+1]$ is a generalized $\delta$-palindrome, which contradicts the maximality of $S[i'\dd j']$.
    \item Finally, if $j > j'$ and $i'=1$, then $i=1$, $j=j'+1$. Hence, $S[i'\dd j']$ obtained from $S[i\dd j]$ by a border-reduction.
    \end{enumerate}
    This completes the proof of the lemma.
  \end{proof}
  
  The combinatorial characterization of Lemma~\ref{lem:ext} yields an algorithm for generating all maximal generalized $d$-palindromes, for all centers and subsequent $d=0,\ldots,\delta$.
  Maximal generalized 0-palindromes are computed using Gusfield's approach (LGPal-queries). For a given $d < \delta$, we consider all the maximal generalized $d$-palindromes and try to extend each of them in all three possible ways (and border-reduce, if possible). This way we obtain a number of generalized $(d+1)$-palindromes amongst which, by Lemma~\ref{lem:ext}, are all maximal generalized $(d+1)$-palindromes. To exclude the non-maximal ones, we group the generalized $(d+1)$-palindromes by their centers (in $\Oh(n)$ time via bucket sort) and retain only the longest one for each center.
We arrive at the following intermediate result.
  
  \begin{lemma}\label{lem:edit}
    Under the edit distance, all maximal generalized $\delta$-palindromes in a string of length $n$ can be computed in $\Oh(n \cdot \delta)$ time and $\Oh(n)$ space.
  \end{lemma}
    
  \section{Maximal palindromic decomposition with gaps and errors}
  Let $\F$ be a set of factors of the text $S[1\dd n]$. In this section we develop a general framework that allows to decompose $S$ into factors from $\F$, allowing at most $g$ gaps. We call such a factorization a $(g,\F)$-factorization of $S$.
  Our goal is to find a $(g,\F)$-factorization of $S$ that minimizes the total length of gaps. We aim at the time complexity $\Oh((n+|\F|) \cdot g)$ and space complexity $\Oh(n\cdot g+|\F|)$.
  
  In our solution we use dynamic programming to compute two arrays, similar to the ones used in Section~\ref{sec:first}:
  \begin{description}
    \item{$MG[1\dd n][0\dd g]$:} $MG[j][q]$ is the minimum total length of gaps in a $(q,\F)$-factorization of $S[1\dd j]$.
    \item{$MG'[1\dd n][0\dd g]$:} $MG'[j][q]$ is the minimum total length of gaps in a $(q,\F)$-factorization of $S[1\dd j]$ for which the position $j$ belongs to a gap.
  \end{description}
  
  We use the following formulas, for $j>0$ and $q>0$:
  \begin{align*}
    MG[j][q]  &= \min(MG'[j][q],\min_{S[a\dd j] \in \F} MG[a-1][q]) \\
    MG'[j][q] &= \min(MG[j-1][q-1],MG'[j-1][q])+1
  \end{align*}
  The border cases are exactly the same as in Section~\ref{sec:first}.
  
  Clearly, the space complexity of this solution is $\Oh(n\cdot g+|\F|)$. Let us analyze its time complexity. Fix $q \in \{0,\ldots,g\}$. The number of transitions using the factors from $\F$ in the dynamic programming is $|\F|$ in total, as each factor is used only for the position $j$ where it ends. Hence, the formulas for $MG[j][q]$ take $\Oh(n\cdot g+|\F|\cdot g)$ time to evaluate. Computing the $MG'[j][q]$ values takes $\Oh(n\cdot g)$ time. Thus we arrive at the desired time complexity of $\Oh((n+|\F|)\cdot g)$.
  
  We apply this approach to maximal generalized $\delta$-palindromes in each of the considered metrics (see the classic result from \cite{Gusfield:1997:AST:262228} for the Hamming distance and Lemma~\ref{lem:edit} for the edit distance) to obtain the following result.
  
  \begin{theorem}\label{thm:main2}
    The \textsc{Generalized Maximal $\delta$-Palindromic Decomposition with Gaps} problem under the Hamming distance or the edit distance can be solved in $\Oh(n\cdot (g+	\delta))$ time and $\Oh(n\cdot g)$ space.
  \end{theorem}
  
\begin{example}
  Consider the following string\footnote{See \texttt{http://www.cesshiv1.org/disview.php?accession=AB220944}} of length 92:

\smallskip
\noindent
\texttt{GGACTCGGCTTGCTGAGGTGCACACAGCAAGAGGCGAGAGCGGCGACTGGTGAGTACGCCAAATTTTGACTAGCGGAGGCTAGAAGGAGAGA}

\smallskip
\noindent
We have used our implementation of the algorithm from Theorem~\ref{thm:main2} to compute the decomposition of the string into maximal complemented 3-palindromes of length at least 14 under the \textbf{edit distance} with at most 4 gaps ($g=4$, $\delta=3$, $m=14$) with the minimal total gap length:

\medskip
\noindent
\texttt{[GGACTCG]\,\underline{G}CTTGCTG\mbox{\underline{\em A}}\underline{G}GTGCAC\underline{A}CAGCAAG\underline{A}\,[GGCGAGAGC]\,GGCGACT\underline{G}\underline{G}\underline{T}\underline{G}AGT\mbox{\underline{\em A}}CGCC\,[AAATTTTG]}

\noindent
\texttt{\underline{A}CTAGC\underline{G}\underline{G}\underline{A}\underline{G}GCTAG\underline{A}\,[AGGAGAGA]
}

\medskip
The gaps are given in square brackets. Edit operations are underlined, with deletes additionally given in italics. The gaps have total length 32.

In comparison, the optimal decomposition of this string under the \textbf{Hamming distance} with the same parameters ($g=4$, $\delta=3$, $m=14$) uses four gaps of total length 46.
\end{example}

  \section{Conclusions}
  We have presented two algorithms for finding palindromic decompositions: one allowing gaps and the other allowing both gaps in the decomposition and errors in palindromes. The first algorithm shows that (somewhat surprisingly) Fici et al.'s algorithm \cite{Fici:2014:SAM:2953214.2953656} for finding an exact palindromic factorization can be extended to handle gaps, a constraint on the palindromes length, and complements in palindromes as well. In the second algorithm we decompose a string into maximal palindromes with errors; the most involved part here was computing all such maximal palindromes under the edit distance.
  
  In the problems that were defined in the beginning, the objective was to minimize the total length of gaps, allowing a certain number of gaps. However, the approaches that were presented in this paper can be used to solve different variants of the problems, like minimizing only the total number of gaps or maximizing the total length of palindromes, regardless of the number of gaps. 
  
  An open question is to efficiently compute decompositions into palindromes that may contain errors and are not necessarily maximal. This problem seems to be hard, as $\delta$-palindromes do not have such a strong combinatorial structure as palindromes without errors.
  
\section*{Acknowledgments}

We would like to warmly thank G.\ Fici, T.\ Gagie, J.\ Kärkkäinen, and D.\ Kempa for providing us Figures~\ref{fig:lemm-gap-properties}, \ref{fig:lemm-gap-properties2}, and \ref{fig:final_lemma} (that we adapted and present in Appendix~\ref{app}) that they first presented in~\cite{Fici:2014:SAM:2953214.2953656}.

\bibliographystyle{plain}
\bibliography{references}

\appendix

\section{Generalized palindromic factorization}\label{app}

In this section we show that the approach of Fici et al.\ \cite{Fici:2014:SAM:2953214.2953656} works for generalized palindromes for any involution $f$. The following auxiliary lemma extends the combinatorial properties of standard palindromes used in \cite{Fici:2014:SAM:2953214.2953656} (see Lemmas 1-3 therein) to generalized palindromes. Recall that a string $y$ is called a \emph{border} of a string $x$ if it is both a prefix and a suffix of $x$. A number $p$ is called a \emph{period} of $x$ if $x[i]=x[i+p]$ for all $i=1,\ldots,|x|-p$. It is well known that $x$ has a period $p$ iff it has a border of length $|x|-p$; see~\cite{AlgorithmsOnStrings,crochemore2003jewels}.

\begin{lemma}\label{lem:aux}
\begin{enumerate}[(a)]
  \item Let $y$ be a suffix of a generalized palindrome $x$. Then $y$ is a border of $x$ iff $y$ is a generalized palindrome.
  \item Let $x$ be a string with a border $y$ such that $| x | \le 2 | y |$. Then $x$ is a generalized palindrome iff $y$ is a generalized palindrome.
  \item Let $y$ be a proper suffix of a generalized palindrome $x$. Then $| x | - | y |$ is a period of $x$ iff $y$ is a generalized  palindrome. In particular, $| x | - | y |$ is the
smallest period of $x$ iff $y$ is the longest generalized palindromic proper suffix of $x$.
\end{enumerate}
\end{lemma}
\begin{proof}
(a) Let $y'$ be the prefix of $x$ of length $|y|$. As $x$ is a generalized palindrome, $y'=f(y^R)$. $(\Rightarrow)$ If $y$ is a border of $x$, then $y = y' = f(y^R)$, so $y$ is a generalized palindrome. $(\Leftarrow)$ If $y$ is a generalized palindrome, then $y' = f(y^R) = y$, so $y$ is a border of $x$.

(b) $(\Rightarrow)$ From (a), if $x$ is a generalized palindrome and $y$ is its border, then $y$ is a generalized palindrome. $(\Leftarrow)$ If $y$ is a generalized palindrome, $f(x^R)$ has a border $f(y^R)=y$. This border covers the whole string $f(x^R)$ and is the same as the border of $x$, so $x=f(x^R)$ and $x$ indeed is a generalized palindrome.

(c) This is a consequence of part (a) and the relation between borders and periods of a string.
\end{proof}

=The crucial combinatorial property of standard palindromes used in Step 1 of the algorithm in Section~\ref{sec:first} is that the sequence of consecutive differences in $P_j$ is non-increasing and contains at most $\Oh(\log j)$ distinct values. We show that the same observation holds for generalized palindromes; this follows from the next lemma, parts (1) and (2). The proof of Lemma~\ref{lem:step1} follows exactly the lines of the proof of the corresponding Lemma~4 in \cite{Fici:2014:SAM:2953214.2953656}; Figures~\ref{fig:lemm-gap-properties} and~\ref{fig:lemm-gap-properties2} (thanks to the courtesy of the authors of~\cite{Fici:2014:SAM:2953214.2953656}) are included for illustration.

\begin{figure}[htpb]
  \centering\small
  \begin{tikzpicture}[scale=0.40]
    \draw[dashed, fill=gray!10] (11,-2) rectangle (14,-1);
    \draw[dashed, fill=gray!10] (0,-1) rectangle (11,0);

    \draw[fill=gray!10] (0,0) rectangle (20,1);
    \draw[fill=gray!10] (11,-1) rectangle (20,0);
    \draw[fill=gray!10] (14,-2) rectangle (20,-1);

    \draw (10, 0.5) node {$x$} rectangle (10, 0.5);
    \draw (15.5, -0.5) node {$y$} rectangle (15.5, -0.5);
    \draw (5.5, -0.5) node {$u$} rectangle (5.5, -0.5);
    \draw (17, -1.5) node {$z$} rectangle (17, -1.5);
    \draw (12.5, -1.5) node {$v$} rectangle (12.5, -1.5);

    \draw[-] (0, 1) .. controls (6, 4) and (14, 4) .. (20, 1);
    \draw[-] (11, 1) .. controls (13.5, 2.5) and (17.5, 2.5) .. (20, 1);
    \draw[-] (14, 1) .. controls (15.5, 2) and (18.5, 2) .. (20, 1);
    
    \draw[-, dotted] (11, 0) -- (11, 1);
    \draw[-, dotted] (14, -1) -- (14, 1);
  \end{tikzpicture}
  \caption{Setting in Lemma~\ref{lem:step1}.}
\label{fig:lemm-gap-properties}
\end{figure}
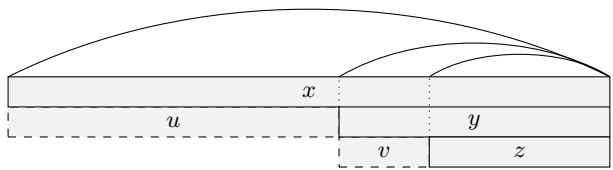

\begin{lemma}\label{lem:step1}
  Let $x$ be a generalized palindrome, $y$ the longest generalized palindromic proper suffix of $x$, and $z$ the longest generalized palindromic proper suffix of $y$. Let $u$
and $v$ be strings such that $x = u y$ and $y = v z$. Then:
\begin{enumerate}[(1)]
\item $| u | \ge | v |$;
\item if $| u | > | v |$ then $| u | > | z |$;
\item if $| u | = | v |$ then $u = v$.
\end{enumerate}
\end{lemma}
\begin{proof}
See Figure~\ref{fig:lemm-gap-properties} for an illustration.

(1) By Lemma~\ref{lem:aux}(c), $| u | = | x | - | y |$ is the smallest period of $x$, and $| v | = | y | - | z |$ is the smallest period of $y$. Since $y$ is a factor of $x$, either $| u | > | y | > | v |$ or $| u |$ is a period of $y$ too, and thus it cannot be smaller than $| v |$.

(2) By Lemma~\ref{lem:aux}(a), $y$ is a border of $x$ and thus $v$ is a prefix of $x$. Let $w$ be a string such that $x = v w$. Then $z$ is a border of $w$ and $| w | = | zu |$ (see Figure~\ref{fig:lemm-gap-properties2}). Since we assume $| u | > | v |$, we must have $| w | > | y |$. Suppose to the contrary that $| u | \le | z |$. Then $| w | = | zu | \le 2 | z |$, and by Lemma~\ref{lem:aux}(b), $w$ is a generalized palindrome. But this contradicts $y$ being the longest generalized palindromic proper suffix of $x$.

\begin{figure}[htpb]
  \centering\small
  \begin{tikzpicture}[scale=0.40]
    \draw[dashed, fill=gray!10] (6,-2) rectangle (9,-1);
    \draw[dashed, fill=gray!10] (0,2) rectangle (3,3);
    \draw[dashed, fill=gray!10] (0,-1) rectangle (6,0);

    \draw[fill=gray!10] (0,0) rectangle (20,1);
    \draw[fill=gray!10] (0,1) rectangle (14,2);
    \draw[fill=gray!10] (3,2) rectangle (14,3);
    \draw[fill=gray!10] (3,3) rectangle (20,4);
    \draw[fill=gray!10] (6,-1) rectangle (20,0);
    \draw[fill=gray!10] (9,-2) rectangle (20,-1);

    \draw (10,0.5) node {$x$} rectangle (10,0.5);
    \draw (7,1.5) node {$y$} rectangle (7,1.5);
    \draw (8.5,2.5) node {$z$} rectangle (8.5,2.5);
    \draw (1.5,2.5) node {$v$} rectangle (1.5,2.5);
    \draw (11.5,3.5) node {$w$} rectangle (11.5,3.5);
    \draw (13,-0.5) node {$y$} rectangle (13,-0.5);
    \draw (3,-0.5) node {$u$} rectangle (3,-0.5);
    \draw (14.5,-1.5) node {$z$} rectangle (14.5,-1.5);
    \draw (7.5,-1.5) node {$v$} rectangle (7.5,-1.5);

    \draw[<->] (14,1.5) -- (20,1.5);
    \draw (17, 1.4) node[above] {$|u|$} rectangle (17, 1.4);
  \end{tikzpicture}
  \caption{Proof of Lemma~\ref{lem:step1}(2): if $|u|>|v|$ and
    $|u|\leq|z|$ then $w$ is a generalized palindromic proper suffix of $x$ longer than $y$.}
\label{fig:lemm-gap-properties2}
\end{figure}
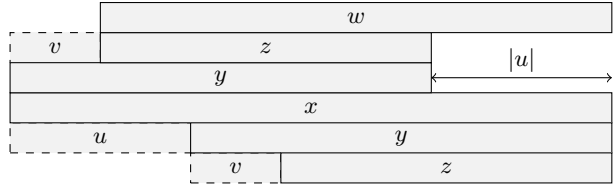

(3) In the proof of (2) we saw that $v$ is a prefix of $x$, and so is $u$ by definition. Thus $u = v$ if $| u | = | v |$.
\end{proof}

We have thus shown that, also in case of generalized palindromes, the set $P_j$ can be compactly represented by a set $G_j$, as described in Section~\ref{sec:first}. To complete Step 1 of the algorithm, we need to show that $G_j$ can be computed from $G_{j-1}$ in $\Oh(\log j)$ time. For this, just as in \cite{Fici:2014:SAM:2953214.2953656}, we show that each triple $(i,\Delta,k) \in G_{j-1}$ will be either eliminated or replaced by $(i-1,\Delta,k)$ in $G_j$. The proof exploits part (3) of Lemma~\ref{lem:step1}.

\begin{lemma}
Let $p_i$ and $p_{i + 1}$ be two consecutive elements of $P_{j - 1 ,\Delta}$. Then $p_i - 1 \in P_j$ iff $p_{i + 1} - 1 \in P_j$.
\end{lemma}
\begin{proof}
By definition, $p_{i + 1} - p_i = \Delta$, and the predecessor of $p_i$ in $P_j$ is $p_{i - 1} = p_i - \Delta$. The strings $x=S[p_{i-1} \dd j-1]$, $y=S[p_{i} \dd j-1]$, and $z=S[p_{i+1} \dd j-1]$ form the situation of Lemma~\ref{lem:step1}(3). Hence, $S[p_{i}-1]=S[p_{i+1}-1]=c$. Thus, $p_i - 1 \in P_j$ iff $S[j]=f(c)$ iff $p_{i + 1} - 1 \in P_j$.
\end{proof}

After this transformation, one might need to update pairs of adjacent triples in $G_j$ because the gaps between them might have changed. This simple process is explained in detail in \cite{Fici:2014:SAM:2953214.2953656} and takes only $\Oh(\log j)$ additional time.

As for Step 2 of the algorithm, it suffices to show that the following combinatorial observation holds for generalized palindromes. Again we follow the lines of the proof from \cite{Fici:2014:SAM:2953214.2953656}; Figure~\ref{fig:final_lemma} (thanks to the courtesy of the authors of~\cite{Fici:2014:SAM:2953214.2953656}) is included for illustration.

\begin{figure}[htpb]
  \centering
  \subfloat[]{
  \centering\small
  \begin{tikzpicture}[scale=0.40]
    \draw[fill=gray!10] (0, 2) rectangle (5, 3);
    \draw[fill=gray!10] (5, 2) rectangle (20, 3);
    \draw[fill=gray!10] (0, -1) rectangle (15, 0);
    \draw[fill=gray!10] (15, -1) rectangle (20, 0);

    \draw (12.5, 2.5) node {$y$} rectangle (12.5, 2.5);
    \draw (0.8, 2) node[below] {$i\hspace{-0.1cm}-\hspace{-0.1cm}\Delta$} rectangle (0.8, 2);
    \draw (0.8, 0) node[above] {$i\hspace{-0.1cm}-\hspace{-0.1cm}\Delta$} rectangle (0.8, 0);
    \draw (14.1, 0) node[above] {$j\hspace{-0.1cm}-\hspace{-0.1cm}\Delta$} rectangle (14.1, 0);
    \draw (3.9, 0) node[above] {$\ell\hspace{-0.1cm}-\hspace{-0.1cm}\Delta$} rectangle (3.9, 0);
    \draw (5.3, 2) node[below] {$i$} rectangle (5.3, 2);
    \draw (19.7, 2) node[below] {$j$} rectangle (19.7, 2);
    \draw (19.7, 0) node[above] {$j$} rectangle (19.7, 0);
    \draw (8.4, 2) node[below] {$\ell$} rectangle (8.4, 2);
    \draw (7.5, -0.5) node {$y$} rectangle (7.5, -0.5);

    \draw[-, dotted, semithick] (8.0, 2) -- (8.0, 3);
    \draw[-, dotted, semithick] (3.0, -1) -- (3.0, 0);

    \draw[-] (0, 3) .. controls (6,7) and (14,7) .. (20, 3);
    \draw[-] (5, 3) .. controls (10,6) and (15,6) .. (20, 3);
    \draw[-, dashed] (8, 3) .. controls (11,5) and (17,5) .. (20, 3);
    \draw[-, dashed] (3, -1) .. controls (6, -3) and (12, -3) .. (15, -1);
    \draw[-] (0, -1) .. controls (5, -4) and (10, -4) .. (15, -1);
  \end{tikzpicture}
  \label{fig:gpl-correctness-1}
  }

  \subfloat[]{
  \centering\small
  \begin{tikzpicture}[scale=0.40]
    \draw[fill=gray!10] (0,0) rectangle (4,1);
    \draw[fill=gray!10] (4,0) rectangle (8,1);
    \draw[fill=gray!10] (8,0) rectangle (15,1);
    \draw[fill=gray!10] (15,0) rectangle (19,1);
    \draw[fill=gray!10] (19,0) rectangle (23,1);

    \draw (2,0.5) node {$w$} rectangle (2,0.5);
    \draw (6,0.5) node {$w$} rectangle (6,0.5);
    \draw (17,0.5) node {$f(w^{R})$} rectangle (17,0.5);

    \draw (21,0.5) node {$f(w^{R})$} rectangle (21,0.5);

    \draw (22.6,0) node[below] {$j$} rectangle (22.6,0);
    \draw (18.15,0) node[below] {$j\hspace{-0.1cm}-\hspace{-0.1cm}\Delta$} rectangle (18.15,0);
    \draw (4.8,0) node[below] {$i\hspace{-0.1cm}-\hspace{-0.1cm}\Delta$} rectangle (4.8,0);
    \draw (1,0) node[below] {$i\hspace{-0.1cm}-\hspace{-0.1cm}2\Delta$} rectangle (1,0);

    \draw[-] (4, 1) .. controls (9,3) and (14,3) .. (19, 1);
    \draw[-] (4, 1) .. controls (10,5) and (17,5) .. (23, 1);
    \draw[-,dashed] (0, 1) .. controls (6,5) and (13,5) .. (19, 1);
  \end{tikzpicture}
  \label{fig:gpl-correctness-2}
  }
  \caption{Proof of Lemma~\ref{lemm:gpl-correctness}. (a) $\ell\in P_j$ iff
    $\ell-\Delta\in P_{j-\Delta}$ for all $\ell\in[i..j]$. (b)
    If $i-2\Delta\in P_{j-\Delta}$
    then $S[i-2\Delta..j]$ is a generalized palindrome.}
    \label{fig:final_lemma}
\end{figure}
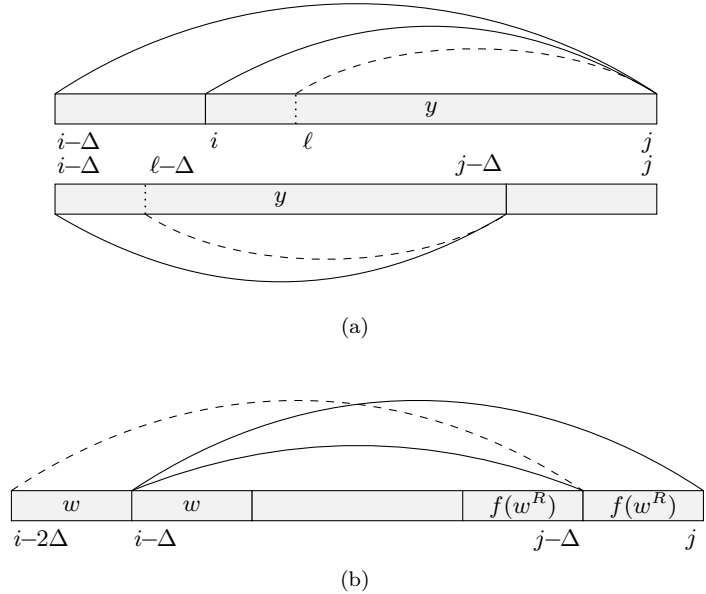

\begin{lemma}\label{lemm:gpl-correctness}
If $(i,\Delta,k) \in G_j$ and $k \geq 2$, then $(i,\Delta,k-1) \in G_{j-\Delta}$.
\end{lemma}
\begin{proof}
By definition, $( i , \Delta, k ) \in G_j$ is equivalent to saying that $P_{j ,\Delta} = \{ i , i + \Delta, \ldots , i + ( k - 1 )\Delta\}$, and we need to show that $P_{j -\Delta,\Delta} = \{ i , i + \Delta, \ldots, i + ( k - 2 )\Delta\}$. We will show first that $P_{j -\Delta,\Delta} \cap [ i - \Delta + 1 \dd j - \Delta] = \{ i , i + \Delta, \ldots, i + ( k - 2 )\Delta\}$ and then that $P_{j -\Delta,\Delta} \cap [ 1 \dd i - \Delta] = \emptyset$.

Since $y = S [ i \dd j ]$ and $x = S [ i - \Delta \dd j ]$ are generalized palindromes and $y$ is the longest proper border of $x$ (by Lemma~\ref{lem:aux}(a)), $S [ i - \Delta \dd j - \Delta] = y = S [ i \dd j ]$. Thus for all $\ell \in [ i \dd j ]$, $\ell \in P_j$ iff $\ell - \Delta \in P_{j -\Delta}$ (see Figure~\ref{fig:gpl-correctness-1}). In particular, the consecutive differences in both cases are the same and for all $\ell \in [ i + 1 \dd j ]$, $\ell \in P_{j ,\Delta}$ iff $\ell - \Delta \in P_{j -\Delta,\Delta}$. Thus $P_{j -\Delta,\Delta} \cap [ i - \Delta + 1 \dd j - \Delta] = \{ i , i + \Delta, \ldots, i + ( k - 2 )\Delta\}$.

We still need to show that $P_{j -\Delta,\Delta} \cap [ 1 \dd i - \Delta] = \emptyset$, which is true if and only if $i - 2 \Delta \not\in P_{j -\Delta}$. Suppose to the contrary that $S [ i - 2 \Delta\dd j -\Delta]$ is a generalized palindrome and let $w = S [ i - 2 \Delta\dd i - \Delta - 1 ]$. Then $S [ j - 2 \Delta + 1 \dd j - \Delta] = f(w^R)$. Since $z = S [ i - \Delta\dd j - \Delta]$ and $S [ i - \Delta\dd j ]$ are generalized palindromes too, we have that $S [ i - \Delta \dd i - 1 ] = w$ and $S [ j - \Delta + 1 \dd j ] = f(w^R)$. Finally, since $z$ is a generalized palindrome, $S [ i - 2 \Delta \dd j ] = w zf(w^R)$ is a generalized palindrome (see Figure~\ref{fig:gpl-correctness-2}). This implies that $i - 2 \Delta \in P_j$ and thus $i - \Delta \in P_{j ,\Delta}$, which is a contradiction.
\end{proof}


\end{document}